\documentclass[preprint,10pt]{elsarticle}

\usepackage[a4paper,scale=0.8]{geometry}
\usepackage{amsmath}
\usepackage{amsthm}
\usepackage{amssymb}
\usepackage{bm}  
\usepackage{url}
\usepackage[numbers]{natbib}
\usepackage{enumerate}
\usepackage{graphicx}
\usepackage{blkarray}
\usepackage{color}
\usepackage{bbm}
\usepackage{dsfont}
\usepackage[toc,page]{appendix}
\usepackage{blkarray}

\newcommand{\conv}{\mathop{\mathrm{conv}}}
\newcommand{\nnegrk}{\mathop{\mathrm{rank}_+}} 

\newcommand{\xc}{\mathop{\mathrm{xc}}} 

\newcommand{\NP}{{\textsc{NP}}}
\newcommand{\PP}{{\textsc{P}}}

\newcommand{\Ppoly}{{\textsc{P/poly}}}

\newcommand{\A}{{\mathtt{A}}}

\newcommand{\I}{{\mathrm{I}}}

\newcommand{\RR}{{\mathbb{R}}}

\newcommand{\EP}{\mathrm{EP}}
\newcommand{\PM}{\mathrm{PM}}
\DeclareMathOperator{\CH}{CH}

\journal{Discrete Applied Mathematics}

\begin{document}
\begin{frontmatter}
\title{Polynomial size linear programs for problems
in $\PP$}

\author[mu]{David Avis\corref{cor1}}
\ead{avis@cs.mcgill.ca}
\cortext[cor1]{Corresponding author}

\author[nu]{David Bremner}
  \ead{bremner@unb.ca}

\author[cu]{Hans Raj Tiwary}
  \ead{hansraj@kam.mff.cuni.cz}

\author[tu]{Osamu Watanabe}
  \ead{watanabe@is.titech.ac.jp}

\address[mu]{GERAD and School of Computer Science, McGill University
   and Graduate School of Informatics,
   Kyoto University}

\address[nu]{Faculty of Computer Science, University of New Brunswick}

\address[cu]{KAM/ITI,
 Charles University}

\address[tu]{Department of Mathematical and Computing Sciences, Tokyo Institute of Technology}

\bibliographystyle{abbrvnat}

\theoremstyle{plain}
\newtheorem{theorem}{Theorem}
\newtheorem{corollary}{Corollary}
\newtheorem{lemma}{Lemma}
\newtheorem{proposition}{Proposition}
\newtheorem{claim}{Claim}
\theoremstyle{definition}
\newtheorem{definition}{Definition}
\newtheorem{example}{Example}
\theoremstyle{remark}
\newtheorem{remark}{Remark}

\BAnewcolumntype{s}{>{$\small$}c}

\newlength{\normalparindent}
\setlength{\normalparindent}{\parindent}
\newlength{\normalparskip}
\setlength{\normalparskip}{\parskip}


\begin{abstract}

A perfect matching in an undirected graph $G=(V,E)$ is a set of vertex disjoint edges from $E$ that include
all vertices in $V$. The perfect matching problem is to decide if $G$ has such a matching.
Recently Rothvo{\ss} proved the striking
result that the Edmonds' matching polytope has exponential extension complexity.
In this paper for each $n=|V|$ we describe
a polytope for the perfect matching problem
that is different from Edmonds' polytope and define a weaker notion of extended formulation. 
We show that the new polytope has a weak extended formulation (WEF) $Q$ of polynomial size. 
For each graph $G$ with $n$ vertices we can readily construct an 
objective function so that solving the resulting linear program
over $Q$ decides whether or not $G$ has a perfect matching. 
With this construction, a straightforward $O(n^4)$ implementation of Edmonds'
matching algorithm using $O(n^2)$ bits of space would yield a WEF $Q$ with $O(n^6 \log n)$ inequalities and variables.
The construction is uniform in the sense that, for each $n$, a single polytope is defined for
the class of all graphs with $n$ nodes.
The method extends to solve polynomial time optimization problems, such as the weighted matching problem.
In this case a logarithmic (in the weight of the optimum solution) number of optimizations are made over
the constructed WEF.

The method described in the paper involves the construction of a compiler that
converts an algorithm given in a prescribed pseudocode into a polytope.
It can therefore be used to construct a polytope for any decision problem
in $\PP$ which can be solved by a well defined algorithm. 
Compared with earlier results of Dobkin-Lipton-Reiss and Valiant our
method allows the construction of explicit linear programs directly
from algorithms written for a standard register model, without
intermediate transformations.
We  apply our results to obtain polynomial upper bounds on the
non-negative rank of certain slack matrices related to membership testing of languages in $\Ppoly$.
\end{abstract}

\end{frontmatter}

\section{Introduction}
\label{intro}
A perfect matching in an undirected graph $G=(V,E)$ is a set of vertex disjoint edges from $E$ that include
all vertices in $V$. 
We let $n$ denote the number of vertices and assume $n$ is even throughout the paper.
The perfect matching problem is to determine if $G$ contains a perfect matching and this can
be decided in polynomial time by running Edmonds' algorithm \cite{Edmonds1965a}.
As well as this combinatorial algorithm, Edmonds also introduced a related polytope
 \cite{Edmonds1965b} which we
will call the {\em Edmonds' polytope} $\EP_n$: 
\begin{equation}
\label{CHEP}
\EP_n =\CH\{ x \in \{0,1\}^{\binom{n}{2}} : x~ \text{is the edge-vector of a perfect matching in }K_n\}
\end{equation}
where $\CH$ is the convex hull operator.

For any $S \subseteq V$ and edge $ij \in E$, we write that $ij \in \delta(S)$ whenever exactly one
of the vertices $i$ and $j$ is in~$S$.
Edmonds \cite{Edmonds1965b} proved that $\EP_n$ has the following halfspace representation:
\begin{eqnarray*}
\label{odd}
\sum_{ij \in \delta(S)} x_{ij} &\geqslant& 1,~~~~~~~~~~~~~S \subseteq V,~|S| \geqslant 3,~|S|~is~odd \\
\sum_{ ij \in \delta(i)} x_{ij} &=& 1 ~~~~~~~~~~~~~i=1,2,...,n\\
0 \leqslant &x_{ij}&  \leqslant ~ 1,~~~~~~ 1 \leqslant i < j \leqslant n
\end{eqnarray*}

This description is exponential in size.
Nevertheless, the perfect matching problem can be solved in polynomial time by solving a linear program (LP)
over this polytope. Indeed, define an objective function $c^Tx = \sum_{ 1 \leqslant i < j \leqslant n} c_{ij} x_{ij}$,
where $c_{ij} = 1$ if $ij \in E$ and $c_{ij} = 0$ otherwise. The LP is:
\begin{eqnarray}
z^* = &\max&~ c^Tx \label{LP1} \\
&s.t.&  x \in \EP_n \nonumber
\end{eqnarray}

It is easy to verify that if $G$ has a perfect matching then $z^* = n/2$ otherwise $z^* < n/2$.
Since the inequalities defining $\EP_n$ can be separated in polynomial time, the LP can be solved
in polynomial time~\cite{GLS}.

Since the perfect matching problem is in $\PP$, it seemed
possible that $\EP_n$ could be written as the projection of a polytope with a polynomial
size description. This is the topic of {\em extension complexity} (see, e.g., Fiorini et al.~\cite{FMPTW15}).
We recall the basic definitions here, referring the reader to \cite{FMPTW15} for further details.

An \emph{extended formulation} (EF) of a polytope $Q \subseteq \RR^q$
is a linear system
\begin{equation} \label{eq:EF}
E x + F y = g,\ y \geqslant \mathbf{0}
\end{equation}
in variables $(x,y) \in \RR^{q+r},$ where $E, F$ are real
matrices with $q, r$ columns respectively, and $g$ is a column vector,
such that $x \in Q$ if and only if there exists $y$ such
that \eqref{eq:EF} holds. The \emph{size} of an EF is defined as the
number of \emph{inequalities} in the system.

An \emph{extension} of the polytope $Q$ is another polytope
$Q' \subseteq \mathbb{R}^e$ such that $Q$ is the image of $Q'$ under a linear map.
We define the \emph{size} of an extension $Q'$ as the number of facets of $Q'$.
Furthermore, we define the \emph{extension complexity} of $Q$, denoted by $\xc{( Q )},$
as the minimum size of any extension of $Q.$

Rothvo{\ss} \cite{Rothvoss17} recently proved the surprising result
that $\xc{(\EP_n)}$ is exponential. Since extension complexity
seemed a promising candidate to obtain computational models that
separate problems in $\PP$ from those that are $\NP$-hard, this was a setback.
A way of strengthening
extension complexity to handle this problem was recently proposed by Avis and Tiwary \cite{AT15a}. 

Dobkin et al.~\cite{DLR79} and Valiant \cite{Valiant82} showed that linear programming 
is $\PP$-complete from which it follows that every problem in $\PP$ has an LP-formulation
of polynomial size.
We will review this result in Section \ref{sect:circuits} giving Valiant's construction.
This construction applies to $\Ppoly$, which is the class of all decision problems 
$L$ solvable by a family 
$C_n, n \ge 1$ of polynomial size Boolean circuits such
that $C_n$ solves the restriction of $L$ to inputs of length $n$.
From these circuits it is straightforward to construct a family of LPs.

The main contribution of this paper is to give a direct method to produce polynomial
size LPs
from polynomial time algorithms, not circuits.
Specifically we will construct LPs directly from a polynomial
time algorithm expressed in pseudocode that solves a decision problem.
Of course a trivial LP formulation can be obtained by first solving the decision problem for a given
input and setting $c=1$ if the answer is yes and $c=0$ otherwise.
Then solving the one dimensional LP: $\max~cx, 0 \leqslant x \leqslant 1$ solves the original problem.
To avoid such trivial LPs we limit how much work can be done in constructing the objective
function. One such limitation might be, for example, to insist that the objective function
can be computed in linear time in terms of the input size. The objective functions we
consider in this paper satisfy this condition.

For concreteness, we focus on
an explicit construction of a polynomial size LP that
can be used to solve the perfect matching problem. Firstly we describe another `natural' 
polytope, $\PM_n$, for the perfect
matching problem. Then we will introduce the notion of a weak extended formulation (WEF). Instead
of requiring projection onto $\PM_n$ we will simply require that LPs solved
over the WEF solve the original problem. The objective function used is basically just 
a $\pm 1$ encoding of the input graph. 
The approach used is quite general and can be applied to any problem in $P$ for which
an explicit algorithm is known. 
It extends to polynomial time solvable optimization problems also. However in this case
a logarithmic (in the weight of the optimum solution) number of optimizations are made over the constructed
WEF.
Note that when an EF exists {\em both} the optimization and decision problems can be solved in a single 
LP optimization. Hence a WEF is weaker since a single LP
solves only the decision problem.
We discuss this further in Section \ref{sect:conclusions}.

The paper is organized as follows. In the next section we introduce a new polytope for the perfect
matching problem and give some basic results about its facet structure.
We define the notion of weak extended formulation and state the main theorem.
In Section \ref{sect:circuits} we first give a simple example to illustrate the technique we use to build
extended formulations from boolean circuits. Then we prove the main theorem of the paper.
In Section \ref{sect:general} we generalize our method to algorithms given in pseudocode
rather than as a circuit. We show how programs written in
a simple pseudocode can be converted to WEFs. Our method is modeled on Sahni's proof of Cook's
theorem given in \cite{HS78}. Since our pseudocode is clearly strong enough to
implement Edmonds' algorithm in polynomial time, our method gives a polynomial size WEF for
the perfect matching problem.
In Section \ref{sect:rank}
we use our main theorem to show that the non-negative rank of certain matrices
is polynomially bounded above.
Finally in Section \ref{sect:conclusions} we give some concluding remarks including
a discussion of applying this technique to polynomial time {\em optimization} problems
such as the maximum weighted matching problem.

\section{Polytopes for decision problems}
\label{sect:poly}
\subsection{Another perfect matching polytope}
  \label{sect:another}

We use the notation $\mathds{1}_{t}$ to denote the $t$-dimensional vector of all ones,
dropping the subscript when it is clear from the context.
Let $n$ be an even integer and let $x$ be a binary vector of length $\binom{n}{2}$.
We let $G(x)=(V,E)$ denote the graph 
with edge incidence vector given by $x$,
let $n$ be the number of its vertices and $m=\mathds{1}^Tx$ the
number of its edges.
Furthermore, let $w_x=1$ if $G(x)$ has a perfect matching and zero otherwise.
We define the polytope $\PM_n$ as:
\begin{equation}
\label{CH}
\PM_n =\CH\{(x,w_x): x \in \{0,1\}^{\binom{n}{2}}\}.
\end{equation}

$\PM_n$ may be visualized by starting with a hypercube in dimension $\binom{n}{2}$ and embedding it in
one higher dimension with extra coordinate $w$. For vertices of the cube corresponding to graphs with
perfect matchings $w=1$ else $w=0$. It is easy to see that $\PM_n$ has precisely $2^{\binom{n}{2}}$
vertices. 
Edmonds' polytope
$\EP_n$ is closely related to $\PM_n$, in fact it defines
one of its faces.
Let
\begin{equation}
\label{EPPM}
F_n =   \PM_n \cap \{  (x,w) : \mathds{1}^T x + (1-w) n^2 = \frac{n}{2}~ \} 
\end{equation}

\begin{proposition}
$F_n$ defines a face of $\PM_n$ and in fact
\begin{equation}
F_n
= \{ (x,1) : x \in \EP_n \}.
\end{equation}
\end{proposition}
\begin{proof}
To show that $F_n$ is a face we show that
the inequality
\begin{equation}
\label{face}
\mathds{1}^T x +(1-w) n^2 \geqslant \frac{n}{2}
\end{equation}
is valid for $\PM_n$. It suffices to verify it for the extreme points
$(x,w_x)$ given in (\ref{CH}).
If $w_x=0$, (\ref{face}) holds since $\mathds{1}^T x +n^2 > \frac{n}{2}$.
Since the inequality is strict, none of these
extreme points lie on $F_n$. 
Otherwise $w_x=1$, $x$ is the incidence vector of graph containing a perfect matching, so
$\mathds{1}^Tx \geqslant n/2$. This shows that $F_n$ defines a face of $\PM_n$.

The vectors $x$ with
$w_x =1$  and $\mathds{1}^Tx=n/2$ are the incidence vectors of perfect matchings of $K_n$
and are precisely those used
to define $\EP_n$ in (\ref{CHEP}). Hence $\EP_n$ lifted by adding the
coordinate $w=1$ is precisely $F_n$.
\end{proof}

For a given input graph $G(\bar{x})=(V,E)$ we define the vector $c=(c_{ij})$ as follows:
\begin{equation}
c_{ij} = \begin{cases}
~~1 & \text{if $ij \in E$}\\
-1 & \text{otherwise}
\end{cases}
~~~~~~~~~~1 \leq i < j \leq n
\label{c}
\end{equation}
and let $d$ be a constant such that $0 < d \leqslant 1/2$.
We construct the LP:
\begin{eqnarray}
\label{LP2}
z^*  = &\max& ~ c^Tx  + d w\\
 &s.t.& (x,w) \in \PM_n \nonumber
\end{eqnarray}

For any positive integer $n$, by an $n$-cube we mean the 
$n$-dimensional hypercube whose vertices are the $2^n$
binary vectors of length $n$.

\begin{proposition}
\label{optlemma}
For any edge incidence vector $\bar{x} \in \{0,1\}^{\binom{n}{2}}$
let $m=\mathds{1}^T\bar{x}$.
The optimum solution to (\ref{LP2}) is unique, $z^* = m +d$ if $G(\bar{x})$ has a perfect matching,
and $z^* = m$ otherwise.
\end{proposition}
\begin{proof}
Let $c$ be the vector defined by (\ref{c}) and
set $m=\mathds{1}^T \bar{x}$.
Note that $c^T \bar{x}=m$ and that $c^T x \leqslant m-1$ for any other vertex $x$
of the $\binom{n}{2}$-cube.
If $G(\bar{x})$ has a perfect matching then $(x,w)=(\bar{x},1)$ 
is a feasible solution to (\ref{LP2}) with $z=m+d$. Since $x \neq \bar{x}$, $c^T x+dw \leqslant m - 1 + d$ and so $(\bar{x},1)$ is the unique
optimum solution.

If $G(\bar{x})$ does not have a perfect matching then $(x,w)=(\bar{x},0)$ 
is a feasible solution to (\ref{LP2}) with $z=m$. Consider any other vertex $x$ of the
$\binom{n}{2}$-cube. Then $z=c^Tx  + d w \leqslant m-1 + 1/2 = m-1/2$.
It follows that
$z^* = m$ and is obtained by the unique solution $(\bar{x},0)$.
\end{proof}

\subsection{Polytopes for decision problems and weak extended formulations}
\label{sect:wef}

The basic ideas above can be extended to arbitrary polynomial time decision problems. 
Let $X$ denote a polynomial time decision problem defined on binary input vectors $x=(x_1,...,x_q)$,
and an additional bit $w_x$, where $w_x=1$ if $x$ results in a ``yes'' answer and $w_x=0$
if $x$ results in a ``no'' answer. We define the polytope $P$ as:
\begin{equation}
\label{CH2}
P =\CH\{(x,w_x): x \in \{0,1\}^q\}
\end{equation}

For a given binary input vector $\bar{x}$ we define the vector $c=(c_{j})$ by:
\begin{equation}
\label{c2}
c_{j} = \begin{cases}
~~1 & \text{if}~\bar{x}_j=1\\
-1 & \text{if}~\bar{x}_j=0
\end{cases}
~~~~~~~~~~1 \leq i < j \leq n
\end{equation}
and let $d$ be a constant such that $0 < d \leqslant 1/2$.
As before we construct an LP:
\begin{eqnarray}
\label{LP3}
z^*  = &\max&~ c^Tx  + d w\\
&s.t.&~  (x,w) \in P \nonumber
\end{eqnarray}
The following proposition can be proved in an identical way to Proposition \ref{optlemma}.
\begin{proposition}
\label{optlemma2}
For any  $\bar{x} \in \{0,1\}^q$
let $m=\mathds{1}^T\bar{x}$.
The optimum solution to (\ref{LP3}) is unique, $z^* = m +d$ if $\bar{x}$ has a ``yes'' answer
and $z^* = m$ otherwise.
\end{proposition}

\begin{definition}
Let $Q$ be a polytope which is a subset of the $(q+t)$-cube
with variables labeled $x_1, ...,x_q,y_1,...,y_t$.
We say that $Q$ has the {\em x-0/1 property} if each of the $2^q$
ways of assigning 0/1 to the variables $x_1, ...,x_q$ uniquely extends to
a vertex $(x,y)$ of $Q$ and, furthermore, $y$ is 0/1 valued. 
$Q$ may have additional fractional vertices.
\end{definition}
In polyhedral terms,
for every binary vector $b \in \RR^q$, the intersection of $Q$ with the
hyperplanes $x_j = b_j$
is a 0/1 vertex.
We will show that we can solve a polynomial time decision problem $X$ 
by replacing $P$ in
(\ref{LP2}) by a polytope $Q$ of polynomial size, while maintaining the same objective 
functions.
We call $Q$ a {\em weak extended formulation}
as it does not
necessarily project onto $P$.

\begin{definition}
\label{def:wef}
A polytope 
\[
Q=\{(x,w,s) : x \in [0,1]^q, w \in [0,1],
s \in [0,1]^r, Ax + bw + Cs \leqslant h\} 
\]
is a {\em weak extended formulation (WEF)} of $P$ if the following hold:
\begin{itemize}
\item
$Q$ has the $x$-0/1 property. 
\item
For any vector $\bar{x} \in \{0,1\}^q$
let $m=\mathds{1}^T\bar{x}$, let $c$ be defined by (\ref{c2})
and let $0 < d \leqslant 1/2$.
If $\bar{x}$ has a ``yes'' answer
the optimum solution of the LP
\begin{equation}
\label{optlp}
z^*=\max~\{c^Tx + dw: (x,w,s) \in Q \}
\end{equation}
is unique and takes the value $z^* = m +d$.
Otherwise $z^* < m+d$.
\end{itemize}
\end{definition}
\noindent
The first condition states that any vertex of $Q$ that has 0/1 values for the
$x$ variables has 0/1 values for the other variables as well. The second
condition connects $Q$ to $P$. For a 0/1 valued vertex $(x,w,s)$ of $Q$
we have $w=1$ if $x$ encodes a ``yes'' answer since $z^*=m+d$. If $x$
encodes a ``no'' answer then $z^* < m+d$ and we must have $w=0$.
The purpose of the coefficient $d$ is so that we can distinguish the two
answers by simply observing the value of $z^*$.

In the sequel we will be concerned with small
positive constants which we will assume are rational and in
reduced form $d=p/q$, for positive integers $p$ and $q$.
The {\em size} of $d$ is $\lceil \log_2 p \rceil  + \lceil \log_2 q \rceil$.

In general $Q$ will have fractional vertices and that is why
the condition for ``no'' answers differs from that given in
Proposition \ref{optlemma2} (an example is given below).
However, for small enough $d$ we can ensure that the LP optimum solution is
unique in both cases and corresponds to that given in
Proposition \ref{optlemma2}.
\begin{proposition}
\label{aboutd}
Let $Q$ be a WEF of $P$. There is a positive constant $d_0$, 
whose size is polynomial in the size of $Q$,
such that for all $d$, $0<d < d_0$, the optimal solution of the LP 
defined in (\ref{optlp})
is unique, $z^* = m +d$ if $\bar{x}$ has a ``yes'' answer
and $z^* = m$ if $\bar{x}$ has a ``no'' answer. 
\end{proposition}
\begin{proof}
The part of the proposition relating to the ``yes'' answers
is already covered by Definition \ref{def:wef}. So consider any 
$\bar{x}$ corresponding to a ``no'' answer and its associated vector $c$
given by (\ref{c2}).
Since $Q$ has the 0/1 property $\bar{x}$ uniquely extends to
a 0/1 valued vertex $(\bar{x},\bar{w},\bar{s})$ of $Q$ with
$\bar{w}=0$ as we saw above.
So for any $d$, $0 < d \le 1/2$, we have $c^T \bar{x} +d\bar{w} = c^T \bar{x} =m$.
Consider any other vertex $(x,w,s)$ of $Q$.
Since $x \neq \bar{x}$, we have
$c^T x < m$. 
Let $d({\bar{x}}) = \min \{m - c^T x \}$ 
taken over all
vertices $(x,w,s)$ with $x \neq \bar{x}$.
This minimum exists, since the number of vertices of $Q$ is finite,
and is positive.
Furthermore,
$c^Tx + d({\bar{x}}) w \le c^Tx + d({\bar{x}}) \le m$.
We define
\begin{equation}
\nonumber
d_0 = \min \{ d(y)~:~ y \in \{0,1\}^q~\text{has a ``no'' answer} \}.
\end{equation}
Now choose any $d$, $0 < d < d_0 $, and any vertex $(x,w,s)$ 
of $Q$ with $x \neq \bar{x}$. If $w>0$ then
$c^Tx + dw < c^Tx + d_0 w \le  c^Tx + d_0 \le m$.
Otherwise $w=0$ and we have $c^Tx + dw = c^Tx < m$ as noted above.
Since the optimal solution for any LP can always be obtained at a vertex,
it follows that the LP defined in (\ref{optlp}) has a unique
optimum solution $z^* = m$ whenever $\bar{x}$ corresponds to a ``no'' answer.

If follows from standard linear programming theory that each vertex $(x,w,s)\in Q$ can
be represented by a number of bits polynomial in the size of $Q$. Since $c$ is a $\pm 1$ vector
and $m \le n$
it follows that $m - c^T x$ can also be so represented. Now $d_0$ is the minimum of
a finite number of such quantities so has size polynomial in the size of $Q$.
The proposition follows.
\end{proof}
This proof gives, at least in principle, a method of computing the constant $d_0$ from
a list of the vertices of $Q$. An easier to compute bound can be obtained from the analysis
of Khachian's ellipsoid
method, see, for example, the text \cite{GLS}. Following this text, let $\langle{}Q\rangle{}$
be the number of bits required to represent $Q$. Then the proof of their Lemma 3.1.33 
implies that the differences $m - c^T x$ in the above proof are all at least $2^{-\langle{}Q\rangle{}}$
and so we could choose this value for $d$. This value will generally be much too small
to be of practical use. 
However, if we are able to observe the {\em value} of $w$ in the optimum solution of (\ref{optlp}) then
we may in fact set $d=0$. In this case 
$z^* = m$ for both answers and it follows from the 0/1 property that the optimum solution is unique
and 0/1 valued. Since $Q$ is a WEF of $P$
the value of $w$ in the optimum solution gives the correct answer.
This is the preferred method in practice as it reduces the problem of 
floating point round off errors
which may be caused by small values of $d$.

We saw in the
proof of Proposition \ref{aboutd}, the reason
that $d$ has to be small is to ensure that fractional vertices
of $Q$
do not lead to optimum solutions for $\pm 1$ valued object
vectors $c$.
As an example, consider the perfect matching problem and the
polytope $\PM_n$ defined earlier.
Suppose $Q_n$ is a WEF for $\PM_n$.
It follows from Propositions \ref{optlemma}
and \ref{aboutd} that
we can determine whether an input graph $G$ has a perfect matching by solving
an LP over either $\PM_n$ or $Q_n$ using the same objective function which
is derived directly from the edge adjacency vector of $G$.
As a very simple case, consider $n=2$ giving $\PM_2 = \CH\{(0,0),(1,1)\}$. 
A possible WEF
with $r=0$ is
given by: 
\[
Q_2= \CH\{(0,0,0),(1,1,1),(1/4,1,1/2)\}
\]
Initially let $d=1/2$.
When $G(\bar{x})$ is an edge, $m=1$, $c_{12}=1$ and $z=c^Tx + dw$
obtains the same
optimum solution of $z^*=3/2=m+d$ over both $\PM_2$ and $Q_2$.
When $G(\bar{x})$ is a non-edge, $m=0$, $c_{12}=-1$ and $z=c^Tx + dw$ 
obtains the optimum solution of $z^* = 0=m$ over $\PM_2$ and $z^* = 1/4<1/2=m+d$ over $Q_2$, at the
fractional vertex (1/4,1,1/2). However, if $0 <d < 1/4$ then $z=c^Tx + dw$
is negative for the fractional vertex and
obtains the unique optimum solution of $z^* = 0=m$ over $Q_2$ and also, of course,
over $\PM_2$.
We see that $Q_2$ projects onto a triangle in the $(x,w)$-space, whereas $\PM_2$
is a line segment. 

In the next section we 
prove the following result:
\begin{theorem}
\label{main_thm}
Every decision problem $X$ in $\Ppoly$ admits a weak extended formulation
$Q$ of polynomial
size.
\end{theorem}

\section{From Circuits to Polytopes}
\label{sect:circuits}

In order to show that Linear Programming is $\PP$-complete,
Valiant~\cite{Valiant82} gave a construction to transform boolean
circuits into a linear sized set of linear inequalities with the
$x$-0/1 property (where $x_i$ are the variables corresponding to the
inputs of the circuit); a similar construction was used by
Yannakakis \cite{Ya91} in the context of the Hamiltonian Circuit
problem. In this section we show that Valiant's construction implies
Theorem~\ref{main_thm}. Valiant's point of view is slightly different
from ours in that he explicitly fixes the values of the input
variables before solving an LP-feasibility problem (as opposed to
using different objective functions with a fixed set of
inequalities). Showing that the result of this fixing is a
$0/1$-vertex is precisely our $x$-0/1 property.

We begin with a standard definition, see for example the text by Arora and Barak~\cite{book:arora-barak}:

\begin{definition}
A {\em (boolean) circuit} with $q$ input bits
$x=(x_1,x_2,...,x_q)$ is a directed acyclic graph in which 
each of its $t$ nodes, called {\em gates}, is either 
an AND ($\wedge$) gate, an 
OR ($\vee$) or a NOT ($\neg$) gate. We label each gate by its output bit. 
One of these gates is designated as the {\em output gate} and gives output bit $w$. 
The {\em size} of a circuit is the number of gates it contains and its 
{\em depth} is the 
maximal length of a path from an input gate to the output gate.
\end{definition}

For example, the circuit shown in Figure \ref{cpm4} can be used to
compute whether or not a graph on 4 nodes has a perfect matching.
The input is the binary edge-vector of the graph and the output is
$w=1$ if the graph has a matching (e.g. $G_1$) or $w=0$ if it does not
(e.g. $G_2$). If the graph has a perfect matching, exactly one of
$y_{12}, y_{13}$ or $y_{14}$ is one, defining the matching.
For each gate we have labeled the output bit by a new variable.
We will construct a polytope from the circuit
by constructing a system of inequalities on the same variables. 

\begin{figure}[!ht] 
  \centering
    \includegraphics[width=0.5\textwidth]{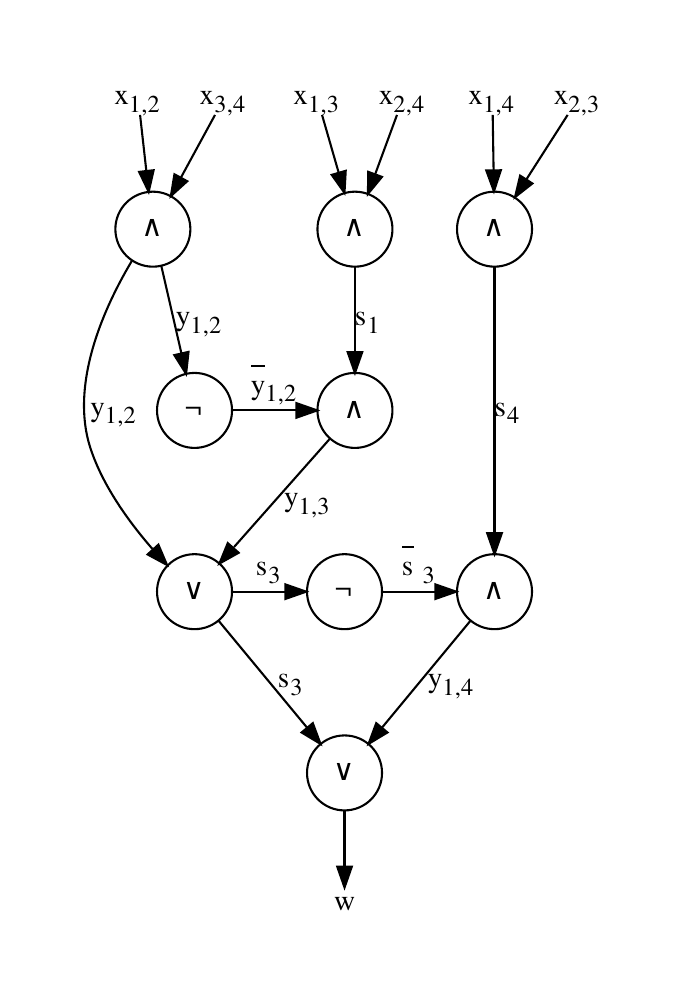}
\caption{A circuit to compute whether a  4 node graph has a perfect matching}
\label{cpm4}
\end{figure}


From an
AND gate, say $ y_{12} = x_{12} \wedge x_{34}$,
we generate the inequalities:
\begin{eqnarray}
x_{12} + x_{34} - y_{12} &\leqslant& 1 \nonumber \\
-x_{12} ~~~~~~~~ + y_{12} &\leqslant& 0  \label{AND}\\
~~~~~~~- x_{34} + y_{12} &\leqslant& 0 \nonumber \\
 y_{12} &\geqslant& 0 \nonumber
\end{eqnarray}

The system (\ref{AND}) defines a polytope in three variables whose 4 vertices
represent the truth table for the AND gate:
\begin{eqnarray*}
&x_{12}&~~ x_{34}~~~~ y_{12}  \\
&0&~~~ 0~~~~~~ 0  \\  
&0&~~~ 1~~~~~~ 0  \\
&1&~~~ 0~~~~~~ 0  \\ 
&1&~~~ 1~~~~~~ 1  
\end{eqnarray*}
Note that the variables $x_{12},x_{34}$ define a 2-cube and so the polytope is an
extension of the 2-cube. In the terminology of the last section, it has
the $\{x_{12},x_{34}\}$-0/1 property.

From an OR gate, say $s_3 = y_{12} \vee y_{13}$,
we generate the inequalities:

\begin{eqnarray}
-y_{12} - y_{13}  + s_3 &\leqslant& 0 \nonumber\\
y_{12} ~~~~~~~ - s_3 &\leqslant& 0 \label{OR} \\
~~~~~~ y_{13} -s_3 &\leqslant& 0 \nonumber \\
s_3 &\leqslant& 1 \nonumber
\end{eqnarray}

The system (\ref{OR}) defines a polytope in three variables whose 4 vertices
represent the truth table for the OR gate, as can easily be checked.
Indeed, this polytope 
has the $\{y_{12},y_{13}\}$-0/1 property. 

From a NOT gate, say $\bar{y}_{12} = \neg y_{12}$, we could generate the equation

\begin{equation}
\bar{y}_{12} = 1 - y_{12} \label{NOT}
\end{equation}

However it is equivalent to just replace all instances of $\bar{y}_{12}$
by $1 - y_{12}$ in the inequality system, and this is what we will do in the
sequel.

The circuit in  Figure \ref{cpm4} contains 5 AND gates and 2 OR gates.
By suitably replacing variables in (\ref{AND}) and (\ref{OR}) we obtain
a system of 28 inequalities in 13 variables. As just mentioned, the NOT
gates are handled by variable substitution rather than explicit equations. 
Let $Q_4$ denote the corresponding polytope.
It will follow by the general argument below
that $Q_4$ is 
a weak extended formulation (WEF) of $\PM_4$.


We now show that the above construction can be applied to any boolean circuit $C$ to obtain
a polytope $Q$ which has the 0/1 property with respect to the inputs of $C$.
\begin{lemma}[\cite{Valiant82}]
\label{01}
Let $C$ be a boolean circuit with $q$ input bits $x=(x_1,x_2,...,x_q)$,
$t$ gates labeled by their output bits $y=(y_1,y_2,...,y_t)$ and
with circuit output bit $w=y_t$.
Construct the polytope $Q$ with $4t$ inequalities and $q+t$ variables
using the systems (\ref{AND}) and (\ref{OR}) respectively.
$Q$ has the the $x$-0/1 property and for every input $x$ the value of $w$ computed
by $C$ corresponds to the value of $y_t$ in the unique extension $(x,y) \in Q$
of $x$.
\end{lemma}
\begin{proof}
Since $C$ is an acyclic directed graph it contains a topological ordering of its
nodes (gates) and we can assume that the labeling $y_1,y_2,...,y_t$ is such
an ordering. Note we can assume $w=y_t$ comes last since it cannot be an input
to any other gate. For any given input $x$ the output of the circuit can
be obtained by evaluating each gate in the order $y_1, ..., y_t$. Since it
is a topological ordering, each input for a gate has been determined before
the gate is evaluated.

We proceed by induction. Let $Q_k$ be the polytope defined by the $4k$
inequalities corresponding to gates $y_1,...,y_k$. The inductive hypothesis is that for
$k=1,2,...,t$ 
\begin{itemize}
\item
$Q_k$ has the $x$-0/1 property, and 
\item
for each $x$ the value of $y_k$ calculated by $C$ corresponds
to the value of $y_k$ in the unique extension of $x$ in $Q_k$.
\end{itemize}
This is clearly true for
$k=1$ as the analysis following (\ref{AND}) and (\ref{OR}) shows. We assume
the hypothesis is true for $k=1,2,...,j$, where $1 \leqslant j < t$, and prove it for $j+1$.
Indeed, since $Q_j$ has the $x$-0/1 property for each $x$ the values of
$y_1,...,y_j$ are uniquely defined and have 0/1 values.
By induction they correspond to the values computed by $C$.
Therefore the analysis following (\ref{AND}) and (\ref{OR}) shows
that $y_{j+1}$ will also be uniquely defined, 0/1 valued, and will correspond to the
value computed by $C$. This verifies the inductive hypothesis for $j+1$
and since $Q=Q_t$ the proof is complete.
\end{proof}

\begin{lemma}
\label{wef}
Let $C$ be a circuit that solves a decision problem $X$ with $q$ input bits $x=(x_1,x_2,...,x_q)$ and has
associated polytope $P$ as defined in (\ref{CH2}). 
The polytope $Q$ constructed in Lemma \ref{01} is a WEF for $P$.
\end{lemma}
\begin{proof}
In order to make the correspondence with Definition \ref{def:wef} we
relabel the variables in $Q$, constructed in Lemma \ref{01}, so
that $s=(y_1, y_2, ..., y_{t-1})$ and $w=y_t$. By  Lemma \ref{01}
we know $Q$ has the $x$-0/1 property so it remains to prove the second
condition in Definition \ref{def:wef}. 

Let $\bar{x}$ be any vector in $\{0,1\}^q$ and set $m=\mathds{1}^T\bar{x}$. 
Since $Q$ has the $x$-0/1 property $\bar{x}$ extends to a unique binary vertex
$(\bar{x}, \bar{w}, \bar{s})$ of $Q$.
Define $c$ as in (\ref{c2}). 
Fix some $d$, $0 < d \leqslant 1/2$ and
consider the optimum solution 
\[
z^*=\max~\{c^Tx + dw: (x,w,s) \in Q \}.
\]
Since $Q$ has the $x$-0/1 property
the maximum of $c^Tx$ over $Q$ is obtained at $c^T\bar{x}=m$ at the unique vertex $(\bar{x}, \bar{w}, \bar{s})$ of $Q$. 
For any other $(x,w,s) \in Q$, since $x$ is in the $q$-cube and not equal to $\bar{x}$,
we have  $c^Tx < m$ and, since $w \leqslant 1$,
$z=c^Tx +dw <m+d$. 
Therefore, if $\bar{x}$ has a ``yes'' answer then $w=1$,  $z^*=m+d$  and $(\bar{x}, \bar{w}, \bar{s})$
is the unique optimum solution. Otherwise $z^*<m+d$.
The lemma follows. 
\end{proof}
Theorem \ref{main_thm} follows from Lemmas \ref{01} and \ref{wef}.
Since there is no limitation of uniformity on the circuits used, the theorem holds
for all decision problems in $\Ppoly$. Since each gate in the circuit
gives rise to 4 inequalities and one new variable, we have the following
corollary.
\begin{corollary}
Let $X$ be a decision problem with corresponding polytope $P$ defined by (\ref{CH2}).
A set of circuits for $X$ with size $p(n)$ 
generates a WEF $Q$ for $P$ with at most $4p(n)$ inequalities and variables.
\end{corollary}


In this section we showed how to construct a polynomial size LP from a polynomial size circuit so that the optimum solution
of the LP gives the output of the circuit. However it is not immediately clear how to use this
to obtain a polytope for the perfect matching problem. 
It would be required to convert Edmonds'
algorithm to a family of circuits. In the next section we bypass this step by showing how to convert
a simple pseudocode directly into a polytope without first computing a circuit 
(See Theorem \ref{thm:pseudocode_to_lp}).
This can be used to convert polynomial time algorithms into polynomial size LPs directly.

 We would like to remark that our construction in Theorem \ref{thm:pseudocode_to_lp} of a WEF 
 from a pseudocode may not be optimal. For example, it would be possible to get roughly $O(T(n)\log T(n))$-
  size circuits simulating a given $T(n)$-time bounded Turing machine
  (see, e.g., Chapter 1 of \cite{book:arora-barak})
  from which we can construct a WEF with $O(T(n)\log T(n))$ inequalities.
  But since Turing Machines are not commonly used for designing algorithms,
  we leave the interested reader to check whether a similar idea
  can be used to define a WEF with smaller size.

\section{Constructing an LP from pseudocode}
  \label{sect:general}
In this section
we introduce a rudimentary pseudocode that can be used for decision problems.
This pseudocode follows the usual practices of specifying algorithms
and the tradition of so-called \emph{register machines} (see
e.g.~\cite{CookReckhow}).
We show how the code can be translated into a linear program, 
in a way similar to that shown for circuits
in the
previous section.
This translation works for any pseudocode, but since the focus of 
this paper is on the class $\PP$,
we will assume there is a polynomial function
$p(n)$ so that the the code terminates within $p(n)$ steps
for any input of size $n$.
In this case we will show that 
the corresponding LP will also have polynomial size in $n$.

The pseudocode we use and its translation into an LP is adapted from a proof of
Cook's theorem given in \cite{HS78} which is attributed to Sartaj Sahni.
In Sahni's construction the underlying algorithm may be non-deterministic, but we
will consider only deterministic algorithms. Furthermore, Sahni describes how
to convert his pseudocode into a satisfiability expression. Although it would be
possible to convert this expression into an LP, considerable simplifications
are obtained by doing a direct conversion from pseudocode to an LP.
In this section, for simplicity,  
we describe only those features of the pseudocode that are
necessary for implementing Edmonds' algorithm for the perfect matching problem.
Additional features would be needed to handle more sophisticated problems, such
as the weighted matching problem. For full details, the reader is referred to
Section 11.2 of \cite{HS78}.

Our pseudocode $\A$ has the following form. We assume $W$ is a fixed integer
which will represent the word size for integer variables.

\begin{itemize}
\item
{\em Variables} are binary valued except for indices, which are 
$W$-bit {\em integers}.
Arrays of binary values are allowed and may be one or two dimensional.
Dimension information is specified at the beginning of $\A$.
We let $q(n)$ denote the maximum number of bits required to represent
all variables for an input size of $n$. Sahni argues that $q(n)=O(p(n))$
however in our case $q(n)$ is significantly smaller.
Statements in $\A$ are numbered sequentially from $1$ to $l$.
\item
An {\em expression} contains at most one boolean operator or is the incrementation
of an index. Array variables are not used in expressions but may be assigned
to simple variables and vice versa.
\item
Certain variables are designated as \emph{parameters} and used to
provide input to the program. All other variables are initially zero.
\item
$\A$ may contain control statements {\bf go to $k$} and
{\bf if $c=1$ then go to $k$ endif}. Here $k$ is an instruction number and $c$ is
a simple binary variable.
\item
$\A$ terminates by setting a binary variable $w$ to one if the input
results in a {\em yes} outcome and to zero otherwise. The program then halts.
\end{itemize}
In our implementation we also allow higher level commands such {\bf while} and {\bf for} loops which
are first precompiled into the basic statements listed above.
As a simple example, here is a pseudocode that produces essentially 
the same result as the circuit in Figure \ref{cpm4}.

\begin{quote}
\begin{quote}
	$y_{12} = x_{12} \wedge x_{34}$ \\
	$y_{13} = 0$ \\
	$y_{14} = 0$ \\
        if $y_{12}$ then go to 50 endif \\
        $y_{13} = x_{13} \wedge x_{24}$ \\
        if $y_{13}$ then go to 50 endif \\
        $y_{14} = x_{14} \wedge x_{23}$ \\
        if $y_{14}$ then go to 50 endif \\
        $w=0$ \\
        return \\
  50:   $w=1$  \\
        return
\end{quote}
\end{quote}
Note that the lines of the pseudocode which are executed depend on the input values $x$.
This is different from the circuit where all gates are executed for every input.
We return to this point below.

The variables in the LP are denoted as follows. They
correspond to variables in $\A$ as it is being executed
on a specific input $\I$. 
\begin{itemize}
\item
{\em Binary variables}  $B(i,t), 1 \leqslant i \leqslant q(n), 0 \leqslant t < p(n)$. \\
$B(i,t)$ represents the value of binary variable $i$ in $\A$ after $t$ steps of computation.
For convenience we may group $W$ consecutive bits together 
as an {\em integer variable} $i$. $I(i,j,t)$ represents the value of the $j$-th bit of integer
variable $i$ in $\A$ after $t$ steps of computation. The bits are numbered from
right to left, the rightmost bit being numbered 1.
\item
{\em Binary arrays} A binary array $R[m],m=0,1,...,u$ is stored in consecutive
binary variables $B(\alpha + m,t), 0 \leqslant m \leqslant u, 0 \leqslant t \leqslant p(n)$ from
some base location $\alpha$. The array index $m$
is stored as a $W$-bit integer $I(*,*,t)$ and so we must have $u \leqslant 2^W -1$.
\item {\em 2-dimensional binary arrays}
A two dimensional binary array $R[m][c]$, $m=0,1,...,u$, $c=0,1,...,v$
is stored in row major order in consecutive binary variables
$B[\alpha+j,t-1]$, $0 \leqslant j \leqslant uv+u+v,  0 \leqslant t \leqslant p(n)$ from
some base location $\alpha$. The array indices $m$ and $c$
are stored as $W$-bit integers $I(*,*,t)$ and so we must have $u,v \leqslant 2^W -1$.

\item
{\em Step counter}  $S(i,t), 1 \leqslant i \leqslant l, 1 \leqslant t \leqslant p(n)$.\\
Variable $S(i,t)$ represents the instruction to be executed
at time $t$. It takes value 1 if line $i$ of $\A$ is being executed at time $t$
and $0$ otherwise.
\end{itemize}

All of the above variables are specifically bounded
to be between zero and one in our LP. The last set of variables, the step counter,
indicates an essential difference between the circuit model and the pseudocode model.
In the former model, all gates are executed for each possible input. The gates can
be executed in any topological order consistent with the circuit. For the pseudocode model,
however, the step to be executed at any time $t$ will usually depend on the actual input.
For each time step $t$ and line $i$ of pseudocode we will 
develop a system of inequalities which have the $x$-0/1 property,
for some subset of variables $x$,
{\em if line $i$ is executed at time $t$}. I.e., the inequalities should uniquely
determine a 0/1 value of all variables given any 0/1 setting of the $x$ variables.
However, if step $i$ is {\em not executed} at time $t$ then the variables should be
free to hold any 0/1 values and these values will be determined by the step that {\em is} executed at time $t$.
So in each set of inequalities a control variable (in our case the variable $S(i,t)$) will appear for this purpose. 
More formally, we make the following definition which generalizes Definition \ref{01}:

\begin{definition} \label{c01}
Let $Cx + Dy \leqslant e$ be a system of inequalities that satisfy the $x$-0/1 property, i.e.
each 0/1 setting of the $x$ variables uniquely defines a 0/1 setting of the $y$ variables.
Suppose that $Cx + Dy \leqslant e+\mathds{1}$ is feasible for all 0/1 settings of $x$ and $y$ variables, and
let $z$ be a binary variable. The system $\mathds{1} z + Cx + Dy \leqslant e +\mathds{1}$
has the ($z$) {\em controlled x-0/1 property}.
\end{definition}
Note that if $z=0$ the new system is always feasible for any 0/1 setting of $x$ and $y$.
If $z=1$ then the new system reduces to the old system that has the $x$-0/1 property.

We now define the 5 different types of linear inequalities needed to simulate the pseudocode which, following Sahni,
we label C,D,E,F and G.\footnote{Sahni also has a constraint set H which
relates to the certificate checking function of his algorithm, and is not needed here.}
Recall that the $S(i,t)$-variables ensure that at each time $t$ a unique line $i$ is executed,
taking the value $S(i,t)=1$ if it is and 0 otherwise. The inequalities listed below 
all have the $S(i,t)$ controlled $x$-0/1 property, and so have the 
form $S(i,t) +  Cx + Dy \leqslant e +1$ for suitably chosen $C,D,e,x,y$.
\begin{itemize}
\item[C:] 
{\em (Variable initialization)} The variables $B(i,0), I(i,j,0), 1 \leqslant i \leqslant q(n), 1 \leqslant j \leqslant W$ 
are set equal to their initial value, if any, else set to zero.
\item[D:]
{\em (Step counter initialization)} Instruction 1 is executed at time $t=1$.
\[
S(1,1)=1 
\]
\item[E:]
{\em (Unique step execution)} A unique instruction is executed at each time $t$.
\[
\sum_{i=1}^l S(i,t)=1,~~~~~~1 \leqslant t \leqslant p(n)
\]
\item[F:]
{\em (Flow control)} The instruction to execute at time $t+1$ is determined,
assuming we are at line $i$ of $\A$ at time $t$, i.e. $S(i,t)=1$. 
If not, i.e. $S(i,t)=0$, then all 
inequalities below are trivially satisfied. This follows since
the other variables are constrained to be between zero and one.
There are 4 subcases
depending on the instruction at line $i$. Inequalities are generated 
for each $t$, $1 \leqslant t \leqslant p(n)$.
\begin{itemize}
\item[(i)] {\bf (assignment statement)} Go to the next instruction. 
\[
S(i,t) - S(i+1,t+1) \leqslant 0~~~~~
\]
\item[(ii)] ({\bf go to $k$})
\[
S(i,t) - S(k,t+1) \leqslant 0~~~~~
\]
\item[(iii)]  ({\bf return}) Loop on this line until time runs out.
\begin{eqnarray*}
S(i,t) - S(i,t+1) \leqslant 0~~~~\\
\end{eqnarray*}
\item[(iv)] ({\bf if $c=1$ then go to $k$ endif}) 
We assume that bit $c$ is represented by variable $B(j,t-1)$.
\begin{eqnarray*}
S(i,t)+B(j,t-1) - S(k,t+1) &\leqslant& 1~~~~~  \\
S(i,t)-B(j,t-1) - S(i+1,t+1) &\leqslant& 0 ~~~~~  
\end{eqnarray*}
\end{itemize}
When $S(i,t)=1$ cases (i)-(iii) fix the next line to be executed and
trivially have the controlled $x$-0/1 property, where $x$ is empty.
For (iv), note we have also the equations E above.
When $S(i,t)=1$, if $B(j,t-1)=1$ then the first inequality fixes $S(k,t+1)=1$ 
otherwise the second inequality fixes $S(i+1,t+1)=1$. The inequalities (iv)
have the controlled $B(j,t-1)$-0/1 property.
 
\item[G:]
{\em (Control of variables)} If we are at line $i$ of $\A$ at time $t$, 
i.e. $S(i,t)=1$, all variables are updated to their correct values at time $t+1$
following the execution of line $i$. 
If not, i.e. $S(i,t)=0$, then all                               
inequalities below are trivially satisfied.
Again there are several cases depending on the instruction at line $i$.
Inequalities are generated    
for each $t$, $1 \leqslant t \leqslant p(n)$.
\begin{itemize}
\item[(i)] {\bf (Reassignment of unchanged variables)} All variables left unchanged
at a given step $t$ need to be reassigned their previous values.
Let $k$ index some bit unchanged at step $t$.
\begin{eqnarray*}
S(i,t)+B(k,t-1)-B(k,t) &\leqslant& 1~~~~~  \\
S(i,t)-B(k,t-1)+B(k,t) &\leqslant& 1~~~~~  
\end{eqnarray*}
Note that when $S(i,t)=1$ these inequalities imply that $B(k,t-1)=B(k,t)$.
They have the controlled $B(k,t-1)$-0/1 property.
Similar inequalities are generated for each integer variable $I(k,j,t), 1\leqslant j \leqslant W$.

{\em In what follows, the above inequalities need to be generated for all
variables $B(k,t)$ and $I(k,j,t)$ not 
being assigned values at time $t$ in the particular instruction $i$ being considered.}
\item[(ii)] ({\bf assignment: $s=x$ and $s= \neg x $})
Assume that $x,s$ are stored in $B(q,t-1),B(s,t)$ respectively.
For $s=x$ we generate the two inequalities:
\begin{eqnarray*}
S(i,t)+B(q,t-1)-B(s,t) &\leqslant& 1~~~~~  \\
S(i,t)-B(q,t-1)+B(s,t) &\leqslant& 1
\end{eqnarray*}
When $S(i,t)=1$ the inequalities imply
$B(s,t)=B(q,t-1)$ as desired. They have the controlled $B(q,t-1)$-property.
For $s = \neg x$ we generate the two inequalities:
\begin{eqnarray*}
S(i,t)+B(q,t-1)+B(s,t) &\leqslant& 2~~~~~  \\
S(i,t)-B(q,t-1)-B(s,t) &\leqslant& 0
\end{eqnarray*}
The analysis is similar to that for $s=x$.

\item[(iii)] ({\bf assignment: $s=x \oplus y$})
Assume that $x,y,s$ are stored in $B(q,t-1), B(r,t-1), B(s,t)$ respectively.
\begin{eqnarray*}
S(i,t)+B(q,t-1)-B(r,t-1)-B(s,t) &\leqslant& 1~~~~~  \\
S(i,t)-B(q,t-1)-B(r,t-1)+B(s,t) &\leqslant& 1~~~~~  \\
S(i,t)-B(q,t-1)+B(r,t-1)-B(s,t) &\leqslant& 1~~~~~  \\
S(i,t)+B(q,t-1)+B(r,t-1)+B(s,t) &\leqslant& 3~~~~~  
\end{eqnarray*}
If $S(i,t)=1$ then all constants on the right hand side are reduced by one and
$S(i,t)$ can be deleted. It is easy to check the inequalities have the
controlled $\{B(q,t-1),B(r,t-1)\}$-0/1 property, and that for each such 0/1 assignment $B(s,t)$
is correctly set.
\item[(iv)] ({\bf assignment: $s=x \wedge y$})
Assume that $x,y,s$ are stored in $B(q,t-1), B(r,t-1), B(s,t)$ respectively.
\begin{eqnarray*}
S(i,t)-B(q,t-1)~~~~~~~~~~~~~~~~+B(s,t) &\leqslant& 1~~~~~  \\
S(i,t)~~~~~~~~~~~~~~~~-B(r,t-1)+B(s,t) &\leqslant& 1~~~~~  \\
S(i,t)+B(q,t-1)+B(r,t-1)-B(s,t) &\leqslant& 2~~~~~  
\end{eqnarray*}
If $S(i,t)=1$ then all constants on the right hand side are reduced by one and
$S(i,t)$ can be deleted. It is easy to check the inequalities have the
controlled $\{B(q,t-1),B(r,t-1)\}$-0/1 property, and that for each such 0/1 assignment $B(s,t)$
is correctly set.
\item[(v)] ({\bf assignment: $s=x \vee y$, $s=x_1 \vee x_2 \vee ... \vee x_k$})\\
Assume that $x,y,s$ are stored in $B(q,t-1), B(r,t-1), B(s,t)$ respectively.
\begin{eqnarray*}
S(i,t)+B(q,t-1)~~~~~~~~~~~~~~~~-B(s,t) &\leqslant& 1~~~~~  \\
S(i,t)~~~~~~~~~~~~~~~~+B(r,t-1)-B(s,t) &\leqslant& 1~~~~~  \\
S(i,t)-B(q,t-1)-B(r,t-1)+B(s,t) &\leqslant& 1~~~~~
\end{eqnarray*}
The analysis is similar to G(iv) and is omitted.
The inequalities have the controlled \{B(q,t-1),B(r,t-1)\}-0/1 property.

The $k$-way {\em or} is an easy generalization
which will be needed in the sequel, where we assume that
$x_j$ is stored in $B(q_j,t-1), j=1,2,...,k$. It is defined
by the following inequalities:
\begin{eqnarray*}
S(i,t)+B(q_j,t-1)-B(s,t) &\leqslant& 1~~~~~1 \leqslant j \leqslant k  \\
S(i,t)-\sum_{j=1}^k B(q_j,t-1) + B(s,t) &\leqslant& 1.
\end{eqnarray*}
\item[(vi)] ({\bf increment integer variable})
Assume that the integer variable is stored in $I(q,j,t-1), 1 \leqslant j \leqslant W$ and
is to be incremented by 1. We require another integer $I(r,j,t), 1 \leqslant j \leqslant W$
to hold the binary carries. On overflow, $I(r,W,t)= 1$ and
$I(q,j,t)=0, 1 \leqslant j \leqslant W$.
The incrementer makes use of two previous operations, G(iii) and G(iv):
\begin{eqnarray*}
I(q,1,t)&=& I(q,1,t-1) \oplus 1  \\
I(r,1,t)&=& I(q,1,t-1) \wedge 1  \\
I(q,j,t)&=& I(q,j,t-1) \oplus I(r,j-1,t)~~~~~2 \leqslant j \leqslant W  \\
I(r,j,t)&=& I(q,j,t-1) \wedge I(r,j-1,t)~~~~~2 \leqslant j \leqslant W  
\end{eqnarray*}
By appropriate
formal substitution of variables, each of the above
assignments is transformed into inequalities of the form G(iii) and G(iv), which are
controlled by the step counter $S(i,t)$. It can be verified that the full system
satisfies the controlled $\{I(q,j,t), 1 \leqslant j \leqslant W\}$-0/1 property because for each 0/1 setting
of these variables all other variables are fixed by the above system of equations.
\item[(vii)] ({\bf equality test for integer variables})
Assume that the integer variables are stored in $I(q,j,t-1)$ and
$I(r,j,t-1)$, $1 \leqslant j \leqslant W$. We require $W+1$ temporary variables w.l.o.g. $B(j,t),1 \leqslant j \leqslant W+1$.
If the two integer variables are equal then 
$B(W+1,t)$ is set to one else it is set to zero.
\begin{eqnarray*}
B(j,t) &=& I(q,j,t-1)\oplus I(r,j,t-1)~~~~~1 \leqslant j \leqslant W  \\
B(W+1,t) &=& \neg  \bigvee_{j=1}^W B(j,t) \\
\end{eqnarray*}
The first equations makes repeated use of G(iii) after appropriate substitution.
By combining G(ii) and the $k$-way {\em or} from $G(v)$ we may implement
the second equation 
by the inequalities.
\begin{eqnarray*}
S(i,t)+B(j,t)+B(W+1,t) &\leqslant& 2~~~~~1 \leqslant j \leqslant W  \\
S(i,t)-\sum_{j=1}^k B(j,t) - B(W+1,t) &\leqslant& 0.
\end{eqnarray*}
The inequalities have the controlled \{$I(q,j,t-1), I(r,j,t-1)$, $1 \leqslant j \leqslant W$\}-0/1 property.

\item[(viii)] ({\bf array assignment})$~~R[m]=x$ (and $x=R[m]$) 
We assume that $R$ 
has dimension $u$, is
stored in $B[\alpha+j,t-1]$, $0 \leqslant j \leqslant u$ and that $x$ is stored in 
$B(x,t-1)$. We further assume that $m$ is stored in
an integer variable $I(m,k,t-1), 1 \leqslant k \leqslant W$.
We need additional binary variables $M(j,t), 0 \leqslant j \leqslant u$ to
hold intermediate results. Initially we write down some equations and then we
use previous results to convert these to inequalities.
Firstly we need to discover the memory location for $R[m]$.
For any $0 \leqslant j \leqslant u$ let $j_W j_{W-1} ... j_1$ be the binary
representation of $j$. Then
we formally define for $k=1,2,...,W$
\[ T_j(m,k,t-1) = \left\{
  \begin{array}{l l}
    I(m,k,t-1) & \quad j_k=0 \\
    1-I(m,k,t-1) & \quad j_k=1
  \end{array} \right.
\]
Note this definition is purely formal and has nothing to do with the execution
of $\A$. We will assign $M(j,t)$ a value via the $W$-way {\em or} given in G(v).
For $0 \leqslant j \leqslant u$:
\begin{eqnarray}
\label{M}
S(i,t) + T_j(m,k,t-1) -M(j,t) &\leqslant& 1~~~~~~~1 \leqslant k \leqslant W \\
S(i,t) - \sum_{k=1}^W T_j(m,k,t-1) +M(j,t) &\leqslant& 1 \notag{}
\end{eqnarray}
When $S(i,t)=1$, it can be verified that
$M(j,t)=0$ whenever $j=m$ and is one otherwise.
Now we may update all array elements of $R$ at time $t$ and make the assignment
$R[m]=x$ by the system of inequalities, for all $0 \leqslant j \leqslant u$:
\begin{eqnarray*}
S(i,t)+B(x,t-1)-B(\alpha + j,t) -M(j,t)&\leqslant& 1~~~~ \\
S(i,t)-B(x,t-1)+B(\alpha + j,t) -M(j,t)&\leqslant& 1~~~~ \\
S(i,t)+B(\alpha + j,t-1)-B(\alpha + j,t) +M(j,t) &\leqslant& 2~~~~ \\
S(i,t)-B(\alpha + j,t-1)+B(\alpha + j,t) +M(j,t) &\leqslant& 2~~~~ \\
\end{eqnarray*}
To understand these inequalities, first note that they are trivially satisfied
unless $S(i,t)=1$. When $j=k$ we have $M(j,t)=0$ and the first two inequalities
are tight. We have $B(\alpha + j,t)=B(x,t-1)$ updating the array element to $x$. 
The second two inequalities
are trivially satisfied. Otherwise $j \neq k$, $M(j,t)=1$, the first two
inequalities are trivially satisfied and the second two are tight.
We have $B(\alpha + j,t)=B(\alpha + j,t-1)$ copying the array element over to time $t$ from time $t-1$.
We remark that there are $O(u W)$ inequalities generated above.

Finally note that we can implement $x=R[m]$ by using the inequalities
\begin{eqnarray*}
S(i,t)+B(x,t)-B(\alpha + j,t-1) - M(j,t) &\leqslant& 1~~~~~ \\
S(i,t)-B(x,t)+B(\alpha + j,t-1) - M(j,t) &\leqslant& 1~~~~~
\end{eqnarray*}
and letting the array $R[m]$ be copied at time $t$ using G(i).
Both of these two inequality systems have the controlled \{$B(x,t-1),B(\alpha + j,t-1),j=0,...,u$\}-0/1 property.

\item[(ix)] ({\bf 2-dimensional array assignment})$~~R[m][c]=x$ (or $x=R[m][c])$.
This is a natural generalization of G(viii).
We assume that $R$
has dimensions $u$ and $v$, is
stored in row major order in 
$B[\alpha+j,t-1]$, $0 \leqslant j \leqslant uv+u+v$ and that $x$ is stored in
$B(x,t-1)$. We further assume that $m$ and $c$ are stored in
an integer variables $I(m,k,t-1), I(c,k,t-1), 1 \leqslant k \leqslant W$ respectively.
We need additional binary variables $M(j,t),0 \leqslant j \leqslant u$ and
$N(j,t), 0 \leqslant j \leqslant v$ to
hold intermediate results. 
Firstly we need to discover the memory location for $R[m][c]$.
We again use the equations (\ref{M}) for the row index.

For the column index,
as in G(viii), for any $0 \leqslant j \leqslant v$ let $j_{W} j_{W-1} ... j_1$ be the binary
representation of $j$. 
We formally define
for $k=1,2,...,W$
\[ T_j(c,k,t-1) = \left\{
  \begin{array}{l l}
    I(c,k,t-1) & \quad j_{k}=0 \\
    1-I(c,k,t-1) & \quad j_{k}=1
  \end{array} \right.
\]
We will assign $N(j,t)$ a value via the $W$-way {\em or} given in G(v).
For $0 \leqslant j \leqslant v$:
\begin{eqnarray*}
\label{N}
S(i,t) + T_j(c,k,t-1) -N(j,t) &\leqslant& 1~~~~~~~1 \leqslant k \leqslant W \\
S(i,t) - \sum_{k=1}^W T_j(c,k,t-1) +N(j,t) &\leqslant& 1
\end{eqnarray*}
When $S(i,t)=1$, it can be verified that
$N(j,t)=0$ whenever $j=c$ and is one otherwise.
Now we may update all array elements of $R$ at time $t$ and make the assignment
$R[m][c]=x$ by the following system of inequalities. For all $0 \leqslant j_1 \leqslant u$,
$0 \leqslant j_2 \leqslant v$, $r=j_1(u+1)+j_2$:
\begin{eqnarray*}
S(i,t)+B(x,t-1)-B(\alpha + r,t) -M(j_1,t)-N(j_2,t) &\leqslant& 1 \\
S(i,t)-B(x,t-1)+B(\alpha + r,t) -M(j_1,t)-N(j_2,t) &\leqslant& 1 \\
S(i,t)+B(\alpha + r,t-1)-B(\alpha + r,t) +M(j_1,t) &\leqslant& 2 \\
S(i,t)-B(\alpha + r,t-1)+B(\alpha + r,t) +M(j_1,t) &\leqslant& 2 \\
S(i,t)+B(\alpha + r,t-1)-B(\alpha + r,t) +N(j_2,t) &\leqslant& 2 \\
S(i,t)-B(\alpha + r,t-1)+B(\alpha + r,t) +N(j_2,t) &\leqslant& 2 \\
\end{eqnarray*}
The analysis is similar to G(viii). The above inequalities are all trivial unless
$S(i,t)=1$. Note that for each $j_1$ and $j_2$, index $r$ gives the relative location in the array.
If $j_1 =m$, $j_2=c$ then $M(j_1,t)=N(j_2,t)=0$, the first two inequalities are tight and
the last four loose. The first two inequalities give $B(\alpha + r,t)=B(x,t-1)$.
Otherwise either $M(j_1,t)=1$ or $N(j_2,t)=1$ or both, and the first two
inequalities are trivially satisfied.
In the former case the two middle inequalities are tight
and we have the equation $B(\alpha + r,t)=B(\alpha + r,t-1)$.
In the latter case this equation is formed from the last two inequalities.
We remark that there are $O(uv + uW + vW)$ inequalities generated above.

For the assignment $x=R[m][c]$ we need the inequalities
\begin{eqnarray*}
S(i,t)+B(x,t)-B(\alpha + r,t-1) - M(j_1,t)-N(j_2,t) &\leqslant& 1 \\
S(i,t)-B(x,t)+B(\alpha + r,t-1) - M(j_1,t)-N(j_2,t) &\leqslant& 1
\end{eqnarray*}
for $r = j_1(u+1)+j_2$. All array elements of $R$ must also be copied
from time $t-1$ to time $t$ as in 
G(i).

Both of these two inequality systems have the controlled \{$B(x,t-1),B(\alpha + j,t),j=0,...,uv+u+v$\}-0/1 property.

Remark: In applications using graphs, a symmetric 2-dimensional array is often 
used to hold the adjacency matrix. Such symmetric matrices may be implemented in pseudocode
by replacing a statement such as $R[m][c]=x$ by the two statements
$R[m][c]=x$ and $R[c][m]=x$. Assignment statements such as $x=R[m][c]$ may be left as is.
\end{itemize}

\end{itemize}

To show the correctness of the above procedure
we give two lemmas that are analogous to Lemmas \ref{01} and \ref{wef} of the last section.
First we show that the above construction can be applied to any pseudocode $\A$,
written in the language described, to produce 
a polytope $Q$ which has the 0/1 property with respect to the inputs of~$\A$.
\begin{lemma}
\label{012}
Let $\A$ be a pseudocode, written in the above language, which takes $n$ input bits $x=(x_1,x_2,...,x_n)$,
and terminates by setting a bit $w$.
Construct the polytope $Q$ as described above relabeling $B(0,p(n))$ as $w$
and the additional variables as $s=(s_1,s_2,...,s_N)$ for some integer $N$.
$Q$ has the $x$-0/1 property and for every input $x$ the value of $w$ computed
by $\A$ corresponds to the value of $w$ in the unique extension $(x,w,s) \in Q$
of $x$.
\end{lemma}
\begin{proof}(Sketch) As with Lemma \ref{01} the proof is by induction, but this time
we use the step counter. By assumption $\A$ terminates after $p(n)$ steps.
Let $k=1,2,...,p(n)$ represent the step counter.
Define $Q_k$ to be the  polytope consisting of precisely those inequalities in $Q$
that use variables: 
$B(i,t), 1 \leqslant i \leqslant q(n), 1 \leqslant t \leqslant k $,
$I(i,j,t), 1 \leqslant i \leqslant q(n), 1 \leqslant j \leqslant W, 1 \leqslant t \leqslant k $ and
$S(i,t), 1 \leqslant i \leqslant l, 1 \leqslant t \leqslant k$. 
\begin{itemize}
\item
$Q_k$ has the $x$-0/1 property, and
\item
for each $x$, at step $k$, $\A$ with input $x$ is
executing line $i$ corresponding to the unique index $i$ where $S(i,k)=1$ 
and all variables at that step have the values corresponding to the values of
$B(i,k), 1 \leqslant i \leqslant q(n)$ and $I(i,j,k), 1 \leqslant i \leqslant q(n), 1 \leqslant j \leqslant W$.
\end{itemize}
The inequalities of $Q_1$ consist of those in groups C, D, and part of E above and
the induction hypothesis is readily verified. We assume
the hypothesis is true for $k=1,2,...,T$, where $1 \leqslant T < p(n)$, and prove it for $T+1$.
Indeed, since $Q_k$ has the $x$-0/1 property for each $x$ the values of
all variables with index $t \leqslant T$ have been correctly set. It follows
that for precisely one index $i$ we have $S(i,T)=1$, meaning that line $i$
of the pseudocode is executed at time $T$ for this particular input.
The inequalities defined in group G all have the controlled $x$-0/1 property
for the control variable $S(i,T)$. 
The variables $B$ and $I$ with index $t=T+1$ are correctly 
set by the analysis in group G above. The analysis in group $F$ implies
that the values of $S(i,T+1)$ will also be uniquely determined and 0/1,
correctly indicating the next line of $\A$ to be executed at $t=T+1$.
This verifies the inductive hypothesis for $T+1$
and since $Q=Q_{p(n)}$ this concludes the proof.
\end{proof}

The next lemma is simply a restatement of Lemma \ref{wef} in the context of
our pseudocode rather than circuits. The proof of Lemma \ref{wef}
makes no reference to how $Q$ was computed, so the same proof holds. 
\begin{lemma}
\label{wef2}
Let $\A$ be a algorithm, written in the above pseudocode, which solves
a decision problem $X$ with $n$ input bits $x=(x_1,x_2,...,x_n)$ and has
associated polytope $P$ as defined in (\ref{CH2}).
The polytope $Q$ described in Lemma \ref{012} is a WEF for $P$.
\end{lemma}

Lemmas \ref{012} and \ref{wef2} justify the correctness of the method
outlined in this section. 

We now analyze the size of the WEF $Q$ created. Recall that $q(n)$ is the
number of bits of storage required by the algorithm A, which of course
consists of a constant number of lines of pseudocode. 
The variables of $Q$ are the variables $B(j,t)$, $S(i,t)$
and additional temporary variables created in some of the groups C--G.
It can be verified that their number is  $O(p(n)q(n))$.
For fixed $t$, each of the
sets of inequalities described in groups C--G have size at most $O(q(n))$ except possibly the array
assignment inequalities described in G(viii) and G(ix). 
As remarked there, an array of dimension $u$
generates $O(u W)$ inequalities. A 2-dimensional array of dimension $r$ by $c$ generates
$O(rc+rW +cW)$ inequalities.  We may assume that $W \in O(\log n)$.
 Then $O(q(n)\log n)$ is
an upper bound the number of inequalities generated in either G(viii) or G(ix).
Since $t$ is bounded by $p(n)$ we see that the WEF has at most $O(p(n)(q(n)\log n))$
inequalities also. We have:
\begin{theorem}\label{thm:pseudocode_to_lp}
Let $X$ be a decision problem with corresponding polytope $P$ defined by (\ref{CH2}).
An algorithm for $X$ written in the pseudocode described above 
requiring $q(n)$ space
and terminating after $p(n)$ steps generates a WEF $Q$ for $P$ with $O(p(n)q(n)\log n)$
inequalities and variables.
\end{theorem}
Since Edmonds' algorithm can be implemented in polynomial time in the pseudocode presented
our method gives a polynomial size WEF
for $\PM_n$.
So for example, a straightforward $O(n^4)$ implementation of Edmonds'
algorithm using $O(n^2)$ space would yield a WEF with $O(n^6 \log n)$ inequalities and variables.
If the $O(n^{2.5})$ time algorithm of Even and Kariv \cite{EK75} can be implemented in our pseudocode
it would yield a considerably smaller polytope. 

{\em Sparktope} is a compiler built along the lines
described above to automatically generate a WEF corresponding to any
given pseudocode \cite{AB17a}. The elements described in C--G above have been implemented
and tested as well as some complete examples of pseudocode. 
The polytopes generated are rather large even for short pseudocodes. For
example, the pseudocode at the beginning of this section generated a 
polytope with about
3200 inequalities!
This should be compared with 28 inequalities for the
circuit in Figure \ref{cpm4} and 4 odd set inequalities for 
Edmonds' polytope $\EP_4$. Nevertheless the WEF generated by this method
should be significantly smaller than $\EP_n$ even for relatively
small $n$.
The details of the implementation of the compiler 
and its application to Edmonds' algorithm
will be described in a subsequent paper.

\section{Connections to non-negative rank} \label{sect:rank}

In this section we reformulate the results in previous sections in terms of
non-negative ranks of certain matrices. Non-negative rank is a very useful tool that captures the extension complexity of polytopes \cite{Ya91}. 
We will prove a necessary condition for membership testing in a language to be in $\Ppoly$ based on the existence
of a certain matrix with polynomially bounded non-negative rank.
This characterization potentially opens the door to proving that a given problem does not 
belong to $\Ppoly$ by demonstrating high non-negative rank of the associated matrix. 
Using a slightly stronger definition of non-negative rank we are also able to give
a sufficient condition for membership testing in a language to be in $\Ppoly$.

A matrix $S$ is called non-negative if
all its entries are non-negative. The non-negative rank of a non-negative
matrix $S$, denoted by $\nnegrk(S)$, is the smallest number $r$ such that there exist non-negative
matrices $T$ and $U$ such that $T$ has $r$ columns, $U$ has $r$ rows and $S=TU$. If we require  the left factor $T$ in the above definition to only contain numbers that can be encoded using a number of bits only polynomial in the number of bits required to encode any entry of $S$ then the smallest such $r$ is called the 
\emph{succinct} non-negative rank of $S$. 
To see the usefulness of this apparently asymmetric restriction on the factors, note that when $S$ is the slack matrix of a polytope $Ax\geqslant b$ then such a factorization allows one to describe an extended formulation for the polytope using only the entries of $A,T$ and $b$. 
So if $T$ is required to be
polynomial in the bit complexity of the entries of $S$ 
then one can represent the polytope as the shadow of another polytope that can be encoded using a polynomial number of bits.

\subsection{Polytopal sandwiches}
Let $P_{\text{in}}$ and $P_{\text{out}}$ be two polytopes in $\mathbb{R}^k$ such
that $P_{\text{in}}\subseteq P_{\text{out}}.$ We say that such a pair defines a
\emph{polytopal sandwich}. With every polytopal sandwich we can associate a
non-negative matrix which encodes the slack of the inequalities defining
$P_{\text{out}}$ with respect to the vertices of $P_{\text{in}}.$ That is, 
if $P_{\text{in}}=\conv(\{v_1,\ldots,v_n\})$ and
$P_{\text{out}}=\{x\in\mathbb{R}^k~:~  
 a_i^T x\leqslant b_i, 1 \leqslant i \leqslant m\}$ then the slack matrix associated
with
the polytopal sandwich thus defined is $S(P_{\text{out}},P_{\text{in}})=S$ with
$S_{ij}=b_i-a_i^T v_j.$  When $P_{\text{in}}$ and $P_{\text{out}}$ define the same polytope $P$ we
denote the corresponding slack matrix simply as $S(P).$ The next lemma shows the relation between the non-negative rank of the slack of a polytopal sandwich and the smallest size polytope whose shadow fits in the sandwich. We will assume that the polytopes defining our sandwiches are full-dimensional. The same lemma appears in \cite{BFPS12} and has roots in \cite{GG10, pashkovich12}. However the proof is simple enough to attribute it to folklore.

\begin{lemma}\label{lem:nnrank_vs_sandwich}
 Let $P_{\text{in}}=\conv(\{v_1,\ldots,v_n\})$ and
$P_{\text{out}}=\{x\in\mathbb{R}^k~:~ 
 a_i^T x\leqslant b_i, 1 \leqslant i \leqslant m\}.$ Let $P_{\text{min}}$ be a polytope with
smallest extension complexity such that 
$P_{\text{in}}\subseteq P_{\text{min}} \subseteq P_{\text{out}}$. Then, $\xc(P_{\text{min}})=
\nnegrk(S(P_{\text{out}},P_{\text{in}})).$
\end{lemma}
\begin{proof}
Let $P$ be any polytope in the sandwich, i.e. $P_{\text{in}}\subseteq P \subseteq P_{\text{out}}.$ We can describe
$P$ as the convex hull of 
 the vertices of $P$ together with the vertices of $P_{\text{in}}.$ Similarly
we can describe $P$ as the 
 intersection of all its facet-defining inequalities and the facet-defining
inequalities of $P_{\text{out}}.$
 Now the matrix $S(P_{\text{out}},P_{\text{in}})$ is a submatrix of the slack
matrix $S(P)$ of this particular 
 representation of $P.$ Therefore,
$\nnegrk(S(P))\geqslant\nnegrk(S(P_{\text{out}},P_{\text{in}})).$  It is 
 easy to see (see, for example, \cite{FMPTW15}) that the non-negative
rank of the slack matrix of a polytope is not changed
 by adding redundant inequalities and points in its representation. Also, since
the non-negative rank of the
 slack matrix of a polytope is equal to its extension complexity
(\cite{FMPTW15}), we have that
 $\xc(P)\geqslant \nnegrk(S(P_{\text{out}},P_{\text{in}})).$
 
 Now, suppose that $\nnegrk(S(P_{\text{out}},P_{\text{in}}))=r.$ That is there
exist non-negative matrices $T$ 
 and $U$ with $r$ columns and rows respectively, such that
$S(P_{\text{out}},P_{\text{in}})=TU.$ Denote by $T_i$ 
 the $i$-th row of $T$ and $U^j$ the $j$-th column of $U$. Consider the
polytope 
 $$Q=\{(x,y)\in\mathbb{R}^{k+r}~:~ a_i^T x + T_iy=b_i, 1 \leqslant i \leqslant m,
y\geqslant 0\}$$ and let 
 $$R=\{x\in\mathbb{R}^k~:~  \exists y\in\mathbb{R}^r, (x,y)\in Q\}.$$
 Since, by definition, $R$ is a projection of $Q$ and $Q$ has at most $r$ inequalities, 
 we have that $\xc(R)\leqslant 
 \nnegrk(S(P_{\text{out}},P_{\text{in}})).$ 

Next we show that $P_{\text{in}}\subseteq R \subseteq P_{\text{out}}$. 
 Suppose $x\in R.$ Then $\exists y, (x,y)\in Q.$ That is $y\geqslant 0$ and
$a_i^T x+T_iy=b_i$ for all $i.$
 Since $T$ is non-negative, $T_iy\geqslant 0$ and therefore $a_i^T x
\leqslant b_i$ for all $i.$ That is,
 $x\in P_{\text{out}}.$ Therefore, $R \subseteq P_{\text{out}}.$
 Suppose $x\in P_{\text{in}}.$ Then $x=\sum_{j=1}^n\lambda_jv_j,
\sum_{j=1}^n\lambda_j=1, \lambda_j\geqslant 0,$
 for some $\lambda.$ Consider $y=\sum_{j=1}^n\lambda_jU^j.$ Then, for each
$i=1,2,..,m$ we have that 
 $a_i^T x+T_iy=\sum_{j=1}^n\lambda_j(a_i^T v_j +
T_iU^j)=\sum_{j=1}^n\lambda_j(b_i)=b_i.$
 Clearly $y\geqslant 0$ since $U$ is non-negative. So $(x,y)\in Q$ and thus
$x\in R.$ Therefore 
 $P_{\text{in}}\subseteq R \subseteq P_{\text{out}}$.

Since $R$ lies in the polytopal sandwich we have
$\xc(R)\geqslant 
 \nnegrk(S(P_{\text{out}},P_{\text{in}}))$ as proved in the first paragraph of the proof. 
Therefore the inequality is in fact an equation and
we may set $P_{\text{min}}=R$ completing the proof.
\end{proof}

Note that given a polytopal sandwich $P_{\text{in}},P_{\text{out}}$, any lower bound on the non-negative rank of its slack matrix $S(P_{\text{in}},P_{\text{out}})$ is also the lower bound on the succinct non-negative rank of $S(P_{\text{in}},P_{\text{out}})$. Conversely, any upper bound on the succinct non-negative rank of $S(P_{\text{in}},P_{\text{out}})$ is also an upper bound on the non-negative rank of $S(P_{\text{in}},P_{\text{out}})$. In the next subsection we will define canonical polytopal sandwiches for binary languages and discuss the relation of the non-negative rank and succinct non-negative rank of the associated slack matrix with whether or not membership testing in the language belongs to $\Ppoly$.

\subsection{Languages and their sandwiches}
In the sequel we consider languages over the $0/1$ alphabet.
Let $L\subseteq\{0,1\}^*$ be a such a language. 
For any positive integer $n$ define the set $L(n)$ as 
$$L(n)=\left\{x\in\{0,1\}^n ~:~x\in L\right\}.$$
To make the connection with Section  \ref{sect:wef}, $L$ plays the role of the decision problem $X$
and $L(n)$ is the set of instances of size $n$ that
have a "yes" answer.

Corresponding to any language $L(n)$ let us define a polytopal sandwich given by a
pair of polytopes. The inner polytope is described by its vertices and is
contained in the outer polytope, which in turn is described by a set of
inequalities. Both the vertices of the inner polytope and the inequalities for
the outer polytope depend only on the language $L(n).$ We call such a sandwich
the \emph{characteristic sandwich} of $L(n)$.
 
For every language $L\subseteq \{0,1\}^*$ we define characteristic
functions 
 $\psi:\{0,1\}^*\to \{0,1\}$ and  $\phi:\{0,1\}^*\to \{-1,1\}^*$ with 
 $$\psi(x)=\left\{\begin{matrix} 1 & ~~~~ \text{ if } x\in L\\ 
 0 & ~~~~ \text{ if } x\notin L\end{matrix}\right.,$$
 
 $$\phi(x)_i=\left\{\begin{matrix} 1 & ~~~~ \text{ if } x_i=1 \\ 
 -1 & ~~~~ \text{ if } x_i=0\end{matrix}\right.$$ 
In terms of Section \ref{sect:wef}, $\psi(x)$ will play
the role of $w_x$
and $\phi(x)$ will play the role of the objective function vector $c$.
The inner polytope is defined by
\begin{equation}
\label{Vdef}
V(n) = CH\{(x,\psi(x))~:~
x\in\{0,1\}^{n}\}. 
\end{equation}
In terms of Section  \ref{sect:wef}, $V(n)$ plays the role of $P$ and
for the perfect matching problem it is $\PM_n$. 

For positive integer $n$ and $0<d \le 1/2$, define
\begin{equation}
\label{Hdef}
H(n,d)=\{(x,w)\in [0,1]^{n+1}~: ~
  \phi(a)^T x + dw \leqslant \mathds{1}^Ta+d\psi(a)~~~
\forall a\in  \{0,1\}^n \}.
\end{equation}
Note that the normal vectors of the inequalities defining $H(n,d)$ are just the
optimization directions
$(c,d)$ that were used in Section \ref{sect:wef}.
 
By direct substitution we see that each vertex $(x,\psi(x))$ of $V(n)$ satisfies the constraints
of $H(n,d)$ and so $V(n) \subseteq H(n,d)$. We will show that the following
matrix is its slack matrix.
For every $n$ consider the $2^n\times 2^n$ non-negative matrix $M(n,d)$ 
defined as follows. Rows and columns of $M(n,d)$ are indexed by $0/1$ vectors 
$a,b$ of length $n$ and
 
 $$ M(n,d)[a,b] = \mathds{1}^Ta - 2a^T b +
\mathds{1}^T b +\alpha(a,b)$$
 
 where $$\alpha(a,b)=\left\{\begin{matrix}d & \text{ if } a\in L(n), b\notin L(n)\\ 
-d & \text{ if }a\notin L(n), b\in L(n)\\ 0 & \text{ otherwise}\end{matrix}\right..$$

 \begin{lemma}
\label{slack}
The slack of $H(n,d)$ with respect to $V(n)$ is the matrix $M(n,d).$  
 \end{lemma}
 \begin{proof}
  Consider two vectors $a,b\in\{0,1\}^n.$ The slack of the inequality
corresponding to $\phi(a)$ with respect to 
  $(b,\psi(b))$ is 
  
  $$\left\{\begin{matrix} 
			  \mathds{1}^Ta+d-d\psi(b)-\phi(a)^Tb &
~~~~ \text{ if } a\in L \\
                          \mathds{1}^Ta
                          -d\psi(b)-\phi(a)^Tb~~~~~~ & ~~~~ \text{ if } a\notin L
          \end{matrix}\right.$$

Observing that $\phi(a)=a-(\mathds{1}-a)=2a-\mathds{1}$ we can see that
$\mathds{1}^Ta-
\phi(a)^Tb=\mathds{1}^Ta+\mathds{1}^Tb-2a^T
b.$ Therefore the slack of the 
inequality corresponding to $\phi(a)$ with respect to $(b,\psi(b))$ is 
$$\left\{\begin{matrix} 
  \mathds{1}^Ta+\mathds{1}^Tb-2a^T b +d  & ~~~~ \text{
if } a\in L,b\notin L\\
  \mathds{1}^Ta+\mathds{1}^Tb-2a^T b-d & ~~~~ \text{
if } a\notin L,b\in L\\
  \mathds{1}^Ta+\mathds{1}^Tb-2a^T b~~~~~ &
\text{otherwise}
        \end{matrix}\right.$$

  \end{proof}
The following lemma is analogous to Proposition \ref{optlemma2}. 
 
 \begin{lemma}
\label{lem:sandwich}
  Let $P$ be a polytope such that $V(n)\subseteq P \subseteq H(n,d).$ Then,
deciding whether a vector 
  $a\in\{0,1\}^n$ is in $L$ or not can be achieved by optimizing over $P$ along
the direction ($\phi(a),d)$.
 \end{lemma}
 \begin{proof}
  Let $a$ be a given vector in $\{0,1\}^n.$ Consider the maxima $z_v, z_p, z_h$ of
$\phi(a)^T x + dw$ when 
  $(x,w)\in V(n)$, $(x,w)\in P,$ and 
$(x,w)\in H(n,d)$ respectively. Since $V(n)\subseteq P \subseteq H(n,d)$ we
have
  that $z_v\leqslant z_p\leqslant z_h.$ 

Since $(a,\psi(a)) \in V(n)$,
$z_v \ge \phi(a)^T a + d \psi(a)=\mathds{1}^Ta+d\psi(a)$.
Furthermore,  
$a$ gives rise to the inequality 
$$
 \phi(a)^T x + dw \leqslant \mathds{1}^Ta+d \psi(a)
$$
in $H(n,d)$ and so
$z_h \le \mathds{1}^Ta +d \psi(a)$. Since $z_v \le z_p \le z_h$, 
$z_v=z_p=z_h=\mathds{1}^Ta+d\psi(a)$. 

Therefore whether
$z_p=\mathds{1}^Ta+d$
  or not tells us whether $a\in L$ or not.
 \end{proof}

 \begin{theorem}  
Let $L$ be a language over the 0/1 alphabet and 
let $n$ be any positive integer. \\
(a) If
$L$ 
belongs to the class $\Ppoly$ then
there is a constant $0 < d < 1/2$ and slack matrix $M(n,d)$ such that both the 
the size of $d$ and
the non-negative rank of  $M(n,d)$ are
polynomial in $n$.  \\
(b) If there is a constant $0< d <1/2$ and slack matrix $M(n,d)$ such that both
the size of $d$ and 
the {\em succinct} non-negative rank of $M(n,d)$ are
polynomial in $n$ then $L$ belongs to the class $\Ppoly$.
  
\label{thm:sandwich}
 \end{theorem}
 \begin{proof}

(a) Suppose $L$ belongs to the class $\Ppoly$ then by
  Lemma \ref{01} there is a polytope $Q$ that is a WEF of 
$V(n)$ with size bounded by a polynomial in $n$. 
Let $P=\{(x,w):(x,w,s) \in Q \}$ be the projection of $Q$ onto its first $n+1$ coordinates.
Since $Q$ is a WEF of $V(n)$ we have $V(n) \subseteq P$.
Applying Proposition \ref{aboutd} to $Q$ we obtain a value $d_0$, whose size is 
polynomially bounded in $n$. Set $d=d_0/2.$
For any vector $a \in \{0,1\}^n$
this proposition implies that if we optimize $\phi(a)x+dw$ over $Q$
the maximum value obtained is $\mathds{1}^Ta+d\psi(a)$. Therefore the inequality corresponding
to this $a$ in the definition of $H(n,d)$ is satisfied for all $(x,w) \in P$. 
It follows that $P \subseteq H(n,d)$ and so $P$ lies in the polytopal sandwich defined
by $V(n)$ and $H(n,d)$.
Since $M(n,d)$ is the 
  slack of $H(n,d)$ with respect to $V(n)$, by Lemma 
  \ref{lem:nnrank_vs_sandwich}
  the non-negative rank of
$M(n,d)$ is upper bounded by the 
extension complexity of any polytope $P$ sandwiched between the two polytopes.
Since the size of $Q$ is bounded above by a polynomial in $n$
and $Q$ is the extension of 
  some polytope that can be sandwiched between $H(n,d)$ and $V(n)$ we have
that the non-negative rank of $M(n,d)$ is bounded by
  a polynomial in $n.$
  
(b) Suppose for some constant $0< d <1/2$ the slack matrix $M(n,d)$ has succinct non-negative rank $r$ 
and that both $d$ and $r$ are
polynomially bounded in $n$. 
Since $M(n,d)$ is the slack of 
  $H(n,d)$ with respect to $V(n),$ by Lemma \ref{lem:nnrank_vs_sandwich}
there exists a polytope $P$ such 
  that the extension complexity of $P$ is $r$ and $V(n)\subseteq P \subseteq
H(n,d).$ In other words, 
  there exist polytopes $P$ and $Q$ such that $Q$ has $r$ facets, projects
to $P$ and  
$V(n)\subseteq P \subseteq H(n,d)$. Furthermore, from the proof of  Lemma \ref{lem:nnrank_vs_sandwich}
we know there exists such a $Q$ whose description contains only numbers polynomially 
bounded in $n$ because the rank $r$ non-negative factorization is succinct by our assumption.
Using the given value of $d$, by Lemma \ref{lem:sandwich} we can optimize
over $P$ to 
  decide whether $x\in L(n)$ or not for a given $x.$ Furthermore, optimizing over
$P$ can be done by optimizing 
  over $Q$ instead, since $Q$ projects to $P$. Since $Q$ has polynomial bit complexity, we can use 
  interior point methods to do the optimization and so determine membership in
$L$ in polynomial time.
 \end{proof}

Theorem \ref{thm:sandwich}(a) in principle paves a way to show that membership 
testing in a certain language is not in $\Ppoly$ by showing that the non negative
rank of associated
slack matrices cannot be polynomially bounded.
Various techniques exist to show lower bounds for the non-negative rank of matrices and have been used to prove 
super-polynomial lower bounds for the extension complexity of important polytopes like the CUT polytope, 
the TSP polytope, and the Perfect Matching polytope of Edmonds, among others \cite{FMPTW15, Rothvoss17, AT15b, PV13}. 
Whether one can apply such techniques to show super-polynomial lower bounds on the non-negative rank 
of the slack of the characteristic sandwich of some language is left as an open problem.

\section{Concluding remarks}
\label{sect:conclusions}


In this paper we have given a method for constructing
polynomial size linear programs for solving decision problems in $\PP$
from their underlying algorithms.
As a concrete example we can construct an LP of size $O(n^6 \log n)$
that decides whether an input graph on $n$ vertices has a perfect matching
or not. This LP is based on an $O(n^4)$ implementation of Edmonds'
algorithm and smaller polytopes can be obtained by using more efficient
implementations. This raises the question of what is the smallest size LP
for the perfect matching problem and whether or not there
is an LP that has an explicit formulation. Along these lines we
mention a recent result in \cite{AT17a} for 2SAT. 
The standard formulation for this problem has superpolynomial
time extension complexity. Applying the method of Section \ref{sect:general}
to a standard 2SAT algorithm yields an LP of size $O(n^4 \log n)$.
However in \cite{AT17a} a simple compact WEF for 2SAT
of size $O(n^3)$ is given.

The discussion in this paper was centred around decision problems 
and one may wonder if it can be applied to optimization problems also.
Before addressing this let us recall some discussion
on the topic in Yannakakis' paper \cite{Ya91}:
\par
\begingroup
\leftskip=2cm
\rightskip\leftskip        
\noindent
Linear programming is complete for
decision problems in $\PP$; the $\PP = \NP ?$ question is equivalent to a weaker requirement
of the LP (than that expressing the TSP polytope), {\em in some sense reflecting the
difference between decision and optimization problems.} (P. 445, emphasis ours)
\par
\endgroup
\vspace{0.5cm}
\noindent
The term ``expressing'' in the citation refers to an extended formulation of polynomial size. 
The method we have described can indeed be used to construct polynomial sized LPs
for optimization problems which have polynomial time algorithms.
Consider, for example, the problem of finding a maximum weight matching for
the complete graph $K_n$,
where the edge weights are integers of length $W$ bits.
For simplicity we assume the weights are non-negative, but a small extension would handle
the general case.
We construct a polytope $P$ as defined by (\ref{CH2}) as follows.
The binary vectors $x$ have length $Wn(n-1)/2 + W +\lceil \log_2 n \rceil $, where the first $Wn(n-1)/2$
bits encode the edge weights and the remaining bits encode an integer $k, 0 \le k < 2^W n$. 
The bit $w_x$ is defined by setting $w_x=1$ whenever the edge weights specified in $x$
admit a matching of weight $k$ or greater and $w_x=0$ otherwise.
Note that the unweighted maximum matching problem for graphs on $n$ vertices
fits into this framework by setting $W=1$.

Applying the method of Section \ref{sect:general} to the weighted version of
Edmonds' algorithm we obtain a polynomial sized WEF $Q$ for $P$. We can
decide by solving a single LP over $Q$ if a given
weighted $K_n$ has a matching of weight at least $k$, for any fixed $k$: 
the last $ W+\lceil \log_2 n \rceil $ coefficients of the objective function (\ref{c2})
vary depending on $k$. 
Therefore, by binary search on $k$ we can solve the maximum weight matching
problem for a given input by optimizing $O(W + \log n)$ times over $Q$ with objective
functions depending on $k$.
We do not, however, know how to solve the weighted matching problem by means of a single
polynomially sized linear program.

\section*{Acknowledgments}
We would like to thank the referees of 
this paper for their detailed comments leading to several important improvements.
This research is supported by 
a Grant-in-Aid for Scientific Research on Innovative Areas -- 
Exploring the Limits of Computation(ELC), MEXT, Japan.
Research of the second author is partially supported by NSERC Canada.
Research of the third author
is supported by project GA15-11559S of GA \v{C}R.

\bibliography{perfect}
\begin{appendices}
\section{Valid inequalities and facets of $\PM_{n}$}
\label{facets}
We give here two classes of valid inequalities for $\PM_n$. Firstly, let $M \subseteq E$  define a perfect matching
in $G$. We have:
\begin{equation}
\label{matchfacet}
w \geqslant \sum_{ij \in M} x_{ij} - \frac{n}{2} + 1
\end{equation}
To see the validity of this inequality, note that if $M$ is a perfect matching the sum becomes $n/2$ and
the inequality states that $w \geqslant 1$, i.e. $G(x)$ has a perfect matching.
We show below
that the each inequality of type (\ref{matchfacet}) define a facet of $\PM_n$.
Since the number of perfect matchings in $K_n$ is the double
factorial $ (n-1)!!=(n-1) \cdot (n-3)...3 \cdot 1$ the number of facet defining
inequalities of $\PM_n$ is therefore super-polynomial.

For a second set of valid inequalities, first let $E_n$ be the set of edges of $K_n$.
A proper subset $S \subset E_n$ is {\em hypo-matchable} if it has no 
matching of size 
$n/2$ but the addition of any other edge from $E_n \setminus S$ to $S$ yields such a matching.
Then we have:
\begin{equation}
\label{hypofacet}
w \leqslant \sum_{ij \in E_n \setminus S} x_{ij}
\end{equation}
To see the validity of this inequality note that if the sum is zero then no edges from $E_n \setminus S$
are in $G(x)$. So $G(x)$ has no perfect matching and so $w$ must be zero.

We now prove that the inequalities (\ref{matchfacet}) are facet defining
for $\PM_{2n}$, where we have replaced $n$ by $2n$ to avoid fractions.
For any integer $s$ we use the notations 
$I_{s \times s}, \mathds{1}_{s \times s}$ and 
$\mathds{O}_{s \times s}$ to represent, respectively, the $s \times s$ identity matrix,
matrix of all ones, and matrix of all zeroes. With only one subscript, the latter two
notations represent the corresponding vectors. For an integer $n$ we let
$t=2n(n-1)$. 
Without loss of generality,
consider a perfect matching $M$ in $K_{2n}$ consisting of the $n$ edges
$12, 34, 56, ..., (2n-1)2n$ and let $E_t$ be the $t$ edges of $K_{2n}$ that are not in
$M$. We construct a set of $t+n=n(2n-1)$ graphs $G(x)$ for which inequality  (\ref{matchfacet})
is tight and for which the $x$ vectors are affinely independent. The corresponding 
$(t+n+1) \times (t+n+1)$-matrix $A$
of edge vectors $x$ is:

\begin{equation}
A =\left [ 
\begin{array}{c|c|c}
I_{t \times t} & \mathds{1}_{t \times n} & \mathds{1}_t \\ 
\cline{1-3}
 \mathds{O}_{n \times t} & \mathds{1}_{n \times n} - I_{n \times n} &  \mathds{O}_{n} \\ 
\cline{1-3}
\mathds{O}_{t} & \mathds{1}_{n}  & 1 
\end{array}
\right ]
\end{equation}
We label the columns of $A$ as follows. The first $t$ columns correspond to the edges
in $E_t$
listed in lexicographical order by $ij$.
The next $n$ columns are indexed by the edges $12, 34, ... , (2n-1)2n$
of $M$ and the final column by $w$. 
The first $t$ rows of $A$ consist of the edge vectors of graphs which contain $M$ and
precisely one other edge $ij$ not in $M$, arranged in lexicographic order by $ij$. 
This means that the top left hand block in $A$ is the
identity matrix. 
Since all these graphs
contain $M$, 
which is a perfect matching, all these remaining entries in the first $t$ rows of $A$ are ones.

The next $n$ rows of $A$ correspond to graphs with edge vectors $M \setminus \{ij\}$,
where $ij$ ranges over the perfect matching $12, 24, ... , (2n-1)2n$. Clearly
the first block of these rows are all zeroes and
the second block is $\mathds{1}_{n \times n} - I_{n \times n}$. 
The last column is all zero since none of these graphs has a perfect matching.
The final row of $A$ corresponds to the graph $M$.

It is straight forward to perform row operations on $A$ to transform it into
an upper triangular matrix with $\pm 1$ on the main diagonal. 
This can be performed by subtracting the last row from the preceding $n$ rows.
The middle block of $A$ is now $-I_{n \times n}$.
Finally these rows can then be added to the last row, which is then divided by $n-1$.
It is then all zero except for the last column, which is -1. This completes the proof.

\end{appendices}
\end{document}